\documentclass[10pt, conference]{ieeeconf}      

\IEEEoverridecommandlockouts                              

\usepackage[utf8x]{inputenc}
\usepackage{amsmath}

\usepackage{amssymb}
\usepackage{amsfonts}
\usepackage{exscale}
\usepackage{graphicx}
\usepackage{booktabs}
\usepackage{authblk}
\usepackage{float}
\usepackage[makeroom]{cancel}
\usepackage{setspace}
\usepackage{soul}
\usepackage{epstopdf}
\makeatletter
\@addtoreset{chapter}{part}
\makeatother  
\usepackage{caption}
\usepackage{subcaption}
\usepackage{kantlipsum,cuted}
\usepackage{mathtools}
\mathtoolsset{showonlyrefs}
\usepackage{cite}

\usepackage{tabularx}

\usepackage{amssymb}
\usepackage{amsfonts}
\usepackage{exscale}
\usepackage{graphicx}
\usepackage{booktabs}
\usepackage{authblk}
\usepackage{float}

\usepackage[makeroom]{cancel}
\usepackage{setspace}
\usepackage{soul}
\usepackage{epstopdf}

\usepackage{caption}
\usepackage{subcaption}
 \usepackage{kantlipsum,cuted}
\usepackage{cite}

\newtheorem{theorem}{Theorem}[section]

\newtheorem{problem}[theorem]{Problem}
\newtheorem{proposition}[theorem]{Proposition}
\newtheorem{corollary}[theorem]{Corollary}

\newtheorem{remark}[theorem]{Remark}
\newtheorem{assumption}[theorem]{Assumption}
\numberwithin{equation}{section}

\DeclareMathOperator{\atantwo}{atan2}

\title{\LARGE \bf Passivity-based design and analysis of phase-locked loops\bf }

\author{Daniele Zonetti$^\star$, Oriol Gomis-Bellmunt, Eduardo Prieto-Araujo, Marc Cheah-Ma\~{n}e\thanks{\noindent All authors are with the Centre d'Innovaci\'{o} Tecnol\`{o}gica en Convertidors Est\`{a}tics i Accionaments (CITCEA-UPC), Departament d'Enginyeria El\`{e}ctrica, Universitat Polit\`{e}cnica de
Catalunya, 08028 Barcelona, Spain. $^\star$Corresponding author: \tt\small daniele.zonetti@upc.edu }}

\begin{document}

\maketitle
\thispagestyle{empty}


\begin{abstract}
    We consider a grid-connected voltage source converter (VSC) and address the problem of estimating the grid angle and frequency—information that is essential for an appropriate operation of the converter. We design phase-locked loop (PLL) algorithms with guaranteed stability properties, under the assumption that the grid is characterized by relatively high short-circuit-ratio (SCR) and inertia. In particular we derive, using passivity arguments, generalization of the conventional synchronous reference frame (SRF) and arctangent (ATAN) PLL, which are the standard solutions in power applications. The analysis is further extended to the case of connection of the VSC to a low-inertia grid, ensuring robustness of the algorithms in case of frequency variations. The results are validated on a benchmark including the grid, the VSC and related controllers.
\end{abstract} 

\section{Introduction}
The ever increasing penetration of renewable energy sources (RES) in the existing power grid is driving dramatic changes in the way power systems are designed and operated. In this continuously evolving scenario, since the integration of RES is enabled by voltage source converters (VSCs),  many questions and challenges arise about their operation and control~\cite{bose}. An appropriate operation of VSCs can be enforced by a variety of controllers, typically consisting of hierarchical control loops operated at different time-scales. The majority of these architectures rely on the definition of a suitable rotating reference frame that is locked to the angle of the voltage at the point of common coupling (PCC)---an information that, unfortunately, is not available to the designer. As a result, an appropriate algorithm that estimates the grid angle must be designed to enforce synchronization of such controllers and then of the converter. \\
Phase-locked loop (PLL) algorithms are the conventional schemes employed to achieve this objective and the analysis of their stability properties in different, perturbed scenarios, has considerably attracted the interest of the power systems and control community\cite{chung2000phase,teodorescu}. To facilitate the analysis, it is usually assumed that PLL algorithms can be operated at a time-scale much slower than the time-scale at which the VSC controllers operate and physical dynamics evolve. Fundamental stability properties of the conventional synchronous reference frame (SRF)-PLL under this assumption have been already studied in~\cite{rantzer} and \cite{abramovitch}. Yet, many questions arise about the ability of the VSCs to synchronize under especially degraded conditions resulting, for example, from the connection of the VSC to grids characterized by low short-circuit-ratio (SCR) \cite{molinasPLL,karimi-weak}, or low inertia~\cite{sun2019impact,rueda}.  However, in all these works only the traditional SRF-PLL mechanism is analysed and no constructive procedure are provided for the design of these schemes. \\
 
In this paper we use passivity arguments~\cite{VAN} to constructively design generalized PLL mechanisms that guarantee---under the aforementioned time-scale separation assumption---a correct estimation of the grid angle and frequency, paving the road to a rigourous stability analysis and potential redesign for grid-connected VSCs in weak and low-inertia grid scenarios. Interestingly enough, specific instances of such schemes correspond to  the conventional SRF-PLL and arctangent (ATAN)-PLL~\cite{miskovic2018linear}, which can be eventually recovered by appropriate selection of some free parameters. Two scenarios are considered. First, we assume that the VSC is connected to a strong grid with a relatively high inertia. Hence, the grid angle can be represented as a balanced three-phase sinusoidal voltage source at constant frequency. Next, we extend our results to the case where the frequency is subject to bounded variations, similar to~\cite{rueda}, where an observer-based framework was employed to analyse low-inertia scenarios. Finally, we validate our results by simulations on a benchmark including the grid, converter and low-level controllers dynamics.\\

The remainder of the paper is organised as follows. Fundamental modeling assumptions and the definition of the estimation problem in suitable coordinates are elaborated in Section~\ref{sec:problem-setup}. Section~\ref{sec:SRF-PLL} and Section~\ref{sec:ATAN-PLL} are then dedicated to the passivity-based design of generalized SRF- and ATAN-PLL---assuming that the grid is characterized by a large inertia. Robustness of these estimation mechanisms to frequency variations is then analysed in Section~\ref{sec:robustness}. The results are validated in Section~\ref{sec:simulations} and next summarized in Section~\ref{sec:conclusions}, where we further provide guidelines for future investigation.\\

\noindent\textbf{Notation.} The set $\mathbb{R}_{>0}$ (respectively $\mathbb{R}_{\geq 0})$ denotes  the set of real positive (respectively non-negative) real numbers. The set $2\mathbb{Z}$ (respectively $2\mathbb{Z}+1$) denotes the subset of even (respectively odd) integers. The set $\mathcal S$ denotes the unit circle in $\mathbb{R}^2$. For given  vectors ${x_1\in\mathbb{R}^{n}}$, $x_2\in\mathbb{R}^{m}$,  $x:=(x_1,x_2)\in\mathbb{R}^{n+m}$ denotes the stacked collection of such vectors and $\|x\|$ denotes its Euclidean norm. \mbox{$\mathrm{Id}:\mathbb{R}\rightarrow\mathbb{R}$} denotes the identity function, \mbox{$\atantwo:\mathcal S\rightarrow [-\pi\;\pi)$} denotes the 2-arguments arctangent function.
For a given angle $\vartheta\in\mathbb{R}$, we define the three-phase vectors:
\begin{equation}\resizebox{1\hsize}{!}{%
$\overline{\sin}(\vartheta):=\begin{bmatrix}
\sin(\vartheta)\\
\sin(\vartheta-\frac{2}{3}\pi)\\
\sin(\vartheta+\frac{2}{3}\pi)
\end{bmatrix}\in\mathbb{R}^3,\quad
\overline{\cos}(\vartheta):=\begin{bmatrix}
\cos(\vartheta)\\
\cos(\vartheta-\frac{2}{3}\pi)\\
\cos(\vartheta+\frac{2}{3}\pi)
\end{bmatrix}\in\mathbb{R}^3.$}
\end{equation}
\section{Problem setup}\label{sec:problem-setup}
In this section we discuss the model considered for the analysis and formulate the problem of estimation of the grid angle and frequency. We make the following assumptions. \vspace{0.05cm}

\begin{assumption}\label{ass:1}
The VSC is characterized by a stable current control scheme, which guarantees instantaneous and exact tracking of the desired currents.
\end{assumption}\vspace{0.05cm}

\begin{assumption}\label{ass:2}
The grid dynamics is characterized by a sufficiently high short-circuit-ratio (SCR).
\end{assumption}\vspace{0.05cm}

Assumption~\ref{ass:1} is usually justified by the fast operation of the current controller as compared to the dynamics of the estimation scheme---a property that can be enforced by appropriate tuning of the gains. Assumption~\ref{ass:2} is verified whenever the grid is characterized by a sufficiently small impedance, so that its dynamics evolve at a time-scale much faster than the time-scale at which the estimation scheme operates.\\
An immediate consequence of these assumptions is that the grid and converter dynamics  can be approximated by their steady-state model and the voltage at the point of common coupling (PCC) can be modeled by a balanced, three-phase, purely sinusoidal signal
\begin{equation}\label{eq:voltage}
    v(t):=V\;\overline{\sin}(\theta(t))\in\mathbb{R}^3,\qquad \theta(t):=\omega t+\phi,
\end{equation}
where $V\in\mathbb{R}$ and $\theta\in\mathbb{R}$ denote the voltage amplitude and angle, with $\omega\in\mathbb{R}_{>0}$ and $\phi\in\mathbb{R}$ denoting the frequency and phase shift respectively. The estimation problem consists then in the design of a sufficiently fast mechanism able to recover the angle and the frequency of \eqref{eq:voltage}, based solely on instantaneous measurements of such a signal. \\
A simple, static approach to solve this problem consists in considering the so-called $\alpha\beta$ transformation\footnote{\label{note:zero}Since the signal \eqref{eq:voltage} is balanced, the zero component can be neglected with no loss of generality.} of \eqref{eq:voltage}, given by:
\begin{equation}
    v_{\alpha\beta}=(V_\alpha,V_\beta)=T_{\alpha\beta}v,\quad T_{\alpha\beta}:=\frac{1}{\gamma}\begin{bmatrix}
    0 &-\frac{\sqrt 3}{2} &\frac{\sqrt 3}{2}\\
    0 &-\frac{1}{2} &-\frac{1}{2}
    \end{bmatrix},
\end{equation}
with $\gamma:=\sqrt{3/2}$, and from which immediately follows:
  \begin{equation}
v_{\alpha\beta}=\gamma V\begin{bmatrix}
  \cos(\theta)\\
  \sin(\theta)
  \end{bmatrix}
  \end{equation}
  The exact value of the angle can be then recovered via:
    \begin{equation}
      \theta=\atantwo(V_\beta,V_\alpha),
  \end{equation}
  and the frequency obtained by computing its first time-derivative.  However, this solution has two major drawbacks that stymie its practical implementation: first, the computation of the angle might be ill-defined for some $v_{\alpha\beta}$; second, frequency estimation relies on exact derivation, an operation that is extremely sensitive to noise and harmonics.    For this reason, dynamic approaches are preferred 
and consist in the design of a (possibly dynamic) frequency estimate $\hat{u}(t)$ for the  estimator:
  \begin{equation}\label{eq:general-estimator}
  \begin{aligned}
      \dot{\hat\theta}(t)&=\hat{u}(t),
\end{aligned}
  \end{equation}
  with $\hat \theta\in\mathbb{R}$, $\hat{\theta}_0\in\mathbb{R}$ denoting the angle estimate and its initial guess respectively. The objective of such design is to ensure that the angle estimation error $\delta:=\hat\theta-\theta\in\mathbb{R}$, captured by the dynamics
  \begin{equation}\label{eq:sys-error}
    \begin{aligned}
        \dot\delta(t)&=-\omega+\hat{u}(t),
    \end{aligned}
\end{equation}
   converges to zero after reasonable transients---a fact that also implies convergence of  $\hat{u}(t)$ to $\omega$. Hence, the grid actual frequency is ultimately provided by $\hat u(t)$. Based on these considerations, we can recast the estimation problem as the following stabilization problem---also referred as \textit{synchronization} problem---where the function $\hat{u}(t)$ can be interpreted as a control signal to be designed.\vspace{0.05cm}
  
    \begin{problem}\label{problem:estimation}
  Consider the system \eqref{eq:sys-error}. Determine $\hat{u}(t)\in\mathbb{R}$   such that 
\begin{equation}
\lim_{t\to\infty}\hat\theta(t)=\theta(t).
\end{equation}
\end{problem}\vspace{0.05cm}

Unfortunately, this problem is complicated by the fact that the signal $\delta\in\mathbb{R}$ is not directly measurable. To overcome this issue, consider the following coordinates change:
  \begin{equation}
  \begin{aligned}
        \hat v_{dq}=(\hat V_d, \hat V_q)=T_{dq}(\vartheta)v,\quad T_{dq}(\vartheta):=\frac{1}{\gamma}\begin{bmatrix}
        \overline{\sin}(\vartheta)\\
        \overline{\cos}(\vartheta)
        \end{bmatrix},
      \end{aligned}
      \end{equation} 
      where $\vartheta:=\hat\theta-\frac{\pi}{2}$, which is referred as the $dq$ transformation of the signal \eqref{eq:voltage} with transformation angle $\vartheta\in\mathbb{R}$. We have then:
      \begin{equation}\label{eq:voltage-dq}
      \hat v_{dq}=
          \gamma V\begin{bmatrix}
      \cos(\delta)\\
      \sin(\delta)
      \end{bmatrix}:=\gamma V  \mathsf{T}(\delta),
      \end{equation}
  where we have  further defined the operator $\mathsf{T}:\mathbb{R}\rightarrow\mathcal S$.
    By calculating the first time-derivative of \eqref{eq:voltage-dq} and using \eqref{eq:sys-error} we obtain the following dynamical system
\begin{equation}
    \begin{aligned}
    \dot{\hat V}_d&=- \dot{\delta}    \gamma V\sin(\delta)=-(\hat{u}-\omega)\hat V_q\\
    \dot{\hat V}_q&=\dot{\delta} \gamma V\cos(\delta)=(\hat{u}-\omega)\hat V_d,
    \end{aligned}
\end{equation}
that in compact form reads:
\begin{equation}\label{eq:sys-dq}
\begin{aligned}
\dot{\hat{v}}_{dq}&=J_2(\hat{u}-\omega)\hat v_{dq}, \quad J_2:=\begin{bmatrix}
0 &-1\\1 &0
\end{bmatrix}\in\mathbb{R}^{2\times 2}.
\end{aligned}
\end{equation}
Note that in contrast with the scalar system \eqref{eq:sys-error}, the state of \eqref{eq:sys-dq} is measurable and defined on the set 
$$
\mathcal{V}:=\gamma V\mathcal S\subset\mathbb{R}^2.$$
We make now the following observation. The set of \textit{assignable} equilibria for this system, \textit{i.e.} the set of equilibria that can be assigned via a constant (equilibrium) control, is given by the whole state space $\mathcal{V}$.  Moreover, for a given $v_{dq}\in\mathcal{V}$, the following relation holds:
\begin{equation}\label{eq:vdqstar}
v_{dq}=\gamma V \mathsf{T}(\delta^h),
\end{equation}
where
$$\delta^h:=\delta_\mathrm{ref}+h\pi,\qquad h\in 2\mathbb{Z},$$ 
for some $\delta_\mathrm{ref}\in[-\pi\;\pi)$.
 As a result, stabilization of $\hat v_{dq}$ to $v_{dq}$ further implies convergence of the angle estimation error $\delta$ to $\delta_\mathrm{ref}$ (modulo $2\pi$). Based on these considerations, Problem \ref{problem:estimation} can be then reformulated as the following stabilization problem in $dq$ coordinates, where correct estimation of the grid angle is imposed by picking a $\delta_\mathrm{ref}=0$.\vspace{0.05cm} 

\begin{problem}\label{problem:dq-stabilization}
Consider the system \eqref{eq:sys-dq} and let 
\begin{equation}
v_{dq\star}:=\gamma V  \mathsf{T}(h\pi)=(\gamma V,0).
\end{equation}
Determine $\hat{u}(t)\in\mathbb{R}$ such that 
\begin{equation}
\lim_{t\to\infty}\hat v_{dq}(t)=v_{dq\star}.
\end{equation}
\end{problem}

\section{SRF-PLL}\label{sec:SRF-PLL}
We now derive a synchronization mechanism with guaranteed stability properties, referred in the sequel as generalized synchronous reference frame (gSRF)-PLL. We have the following preliminary proposition. \vspace{0.05cm}

\begin{proposition}\label{prop:passivity-1}
Consider the system \eqref{eq:sys-dq},  let $v_{dq}\in\mathcal{V}$ and define
\begin{equation}\label{eq:incremental-1}
     \tilde{u}:=\hat{u}-\omega,\quad \tilde y_1:=-v_{dq}^\top J_2\hat v_{dq}.
\end{equation}
Then the  map $\tilde{u}\rightarrow \tilde y_1$   is passive with storage function 
\begin{equation}\label{eq:energy-H1}
    \mathcal{H}_{1}(\hat v_{dq}):=\frac{1}{2}
    \|\hat v_{dq}-v_{dq}\|^2.
\end{equation}
\end{proposition}\vspace{0.05cm}

\begin{proof}
Consider the energy function \eqref{eq:energy-H1}. By calculating its derivative along the system's trajectories we get:
\begin{equation}
\begin{aligned}
    \dot{\mathcal{H}}_{1}&=(\hat{v}_{dq}-v_{dq})^\top \dot{\hat v}_{dq}=
    (\hat{v}_{dq}-v_{dq})^\top  J_2\tilde{u}\hat v_{dq }=
    \tilde y_1\tilde{u}
    \end{aligned}
\end{equation}
    where in the second equivalence we simply replaced \eqref{eq:sys-dq}, while the last one follows from skew-symmetry of $J_2$ the definition of $\tilde y_1$, thus completing the proof.
\end{proof}
\vspace{0.05cm}

The passive property of the system can be then employed to constructively design  a class of estimation algorithms in the form of PI controllers driven by the ouput $\tilde y_1$, as explained in the following proposition.\vspace{0.05cm}

\begin{proposition}\label{prop:SRF-PLL-dq}
Let $v_{dq } \in\mathcal{V}$. Consider the system \eqref{eq:sys-dq} in closed-loop with the PLL:
\begin{equation}\label{eq:gSRF-PLL-dq}
\begin{aligned}
    \hat{u}&=-K_P\Phi(\tilde y_1)+\hat{\omega}\\
    \dot{\hat{\omega}} &=-K_I\tilde y_1,
    \end{aligned}
\end{equation}
with state $\hat{\omega}\in\mathbb{R}$, $\tilde y_1$ given by \eqref{eq:incremental-1}, $K_I$, $K_P$ positive scalars and $\Phi$ is a scalar function verifying 
$\Phi(0)=0$, $\Phi(s)s>0$ for any $s\neq 0$.  Then for any initial condition $(v_{dq0},\omega_0)\in\mathcal{V}\times\mathbb{R}$ the system's trajectories converge to the set
$$\mathcal{E}_{dq}:=(v_{dq },\omega)\cup(-v_{dq },\omega).$$
\end{proposition}\vspace{0.05cm}

\begin{proof} First, note that at the equilibrium $\hat{\omega}=\omega.$ Then, consider the augmented energy function
    \begin{equation}
        \mathcal W_{1}(\hat v_{dq},\hat{\omega}):=\mathcal{H}_{1}(\hat v_{dq})+\frac{1}{2K_I}(\hat{\omega}-\omega)^2,
    \end{equation}
    which is positive definite and radially unbounded, being $K_I$ positive. Then
    \begin{equation}
    \begin{aligned}
        \dot{\mathcal{W}}_{1}&=\dot{\mathcal{H}}_{1}+\frac{1}{K_I}(\hat{\omega}-\omega)\dot{\hat{\omega}}\\
        &=\tilde y_1(-K_P\Phi(\tilde y_1)+\hat{\omega}-\omega) -(\hat{\omega}-\omega) \tilde y_1\\
        &=-K_P \Phi(\tilde y_1)\tilde y_1\\
        &\leq 0,
        \end{aligned}
    \end{equation}
    where in the second equivalence we replaced the dynamics of \eqref{eq:gSRF-PLL-dq} and used Proposition~\ref{prop:passivity-1}, while the last inequality results from positivity of $K_P$ and the definition of $\Phi$. The proof is completed using LaSalle invariance principle and noting that $\dot{\mathcal{W}}_{1}=0$ implies $(\hat v_{dq},\hat{\omega})= (\pm v_{dq },\omega)$.
\end{proof}\vspace{0.05cm}

Proposition \ref{prop:SRF-PLL-dq} merely states that for any initial guess of the gSRF-PLL, the trajectories converge to the set $\mathcal{E}_{dq}$, which consists of two isolated equilibria $(\pm v_{dq },\omega)\in\mathcal{V}\times\mathbb{R}$. We deduce then that the frequency $\omega$ is correctly estimated via \eqref{eq:gSRF-PLL-dq}. It is also  immediate to see that a specific implementation of the controller \eqref{eq:gSRF-PLL-dq} for $\Phi:=\frac{1}{\gamma V}\mathrm{Id}$ and $v_{dq }=v_{dq\star}$ is given by the conventional SRF-PLL, as stated in the following corollary.\vspace{0.05cm}    

 \begin{corollary}\label{corollary:SRF-PLL-dq} Consider the system \eqref{eq:sys-dq} in closed-loop with the SRF-PLL:
\begin{equation}\label{eq:SRF-PLL-dq}
\begin{aligned}
    \hat{u}&=-K_P\hat V_q+\hat{\omega}\\
    \dot{\hat{\omega}} &=-K_I\hat V_q,
    \end{aligned}
\end{equation}
with $K_I$, $K_P$ positive scalars. 
 Then for any initial condition  $(v_{dq0},\omega_0)\in\mathcal{V}\times\mathbb{R}$ the system's trajectories converge to the set
$$\mathcal{E}_{dq\star}:=(v_{dq\star},\omega)\cup(-v_{dq\star},\omega).$$
\end{corollary}\vspace{0.05cm}

It must be noted that Corollary~\ref{corollary:SRF-PLL-dq} is inconclusive about the ability of the SRF-PLL to correctly estimate the angle. Indeed, while such objective is achieved if trajectories converge to $(v_{dq\star},\omega)$, convergence to $(-v_{dq\star},\omega)$ would result into an error of $\pi$ radians. To show that the latter is an exceptional case, let us rewrite the SRF-PLL \eqref{eq:SRF-PLL-dq} in polar coordinates using \eqref{eq:voltage-dq}, leading to the following nonlinear system:
  \begin{equation}\label{eq:SRF-PLL}
      \begin{aligned}
           \dot{\delta}& =-\omega-K_P\gamma V\sin(\delta)+\hat \omega\\
          \dot{\hat{\omega}} &= -K_I\gamma V\sin(\delta).
      \end{aligned}
  \end{equation}
By setting the left-hand side to zero, it is immediate to see that in such coordinates the system has infinitely many equilibria $\mathcal{E}_\star:=\cup_n(n\pi,\omega)$, with $n\in\mathbb{Z}$, and that these are mapped to the elements of $\mathcal{E}_{dq\star}$. More precisely, let us define
    \begin{align*}
        \mathcal{E}_\star^{+}&:=\cup_h(h\pi,\omega),\quad h\in 2\mathbb{Z} \\
        \mathcal{E}_\star^{-}&:=\cup_k(k\pi,\omega),\quad k\in 2\mathbb{Z}+1.
    \end{align*}
    Therefore, we have:
    \begin{equation}\label{eq:map-equilibria}
    \begin{aligned}
    \mathcal{E}_{dq\star}^+:&=(+v_{dq\star},\omega)=(\gamma V  \mathsf{T}(\delta^h),\omega),\quad 
    \forall\delta^h
    \in\mathcal{E}_\star^{+}\\
    \mathcal{E}_{dq\star}^-:&=(-v_{dq\star},\omega)=(\gamma V  \mathsf{T}(\delta^k),\omega),\quad 
    \forall\delta^k\in\mathcal{E}_\star^{-}.
    \end{aligned}
    \end{equation}  We are now ready for the following propositions, similar to~\cite{garces}.\footnote{The same result apply, with minor differences, to the gSRF-PLL \eqref{eq:gSRF-PLL-dq}.}\vspace{0.05cm}

\begin{proposition}\label{prop:SRF-linearized}
Any element of $\mathcal{E}_\star^{-}$ is a saddle equilibrium point for the system \eqref{eq:SRF-PLL}.
\end{proposition}\vspace{0.05cm}

\begin{proof}
  For the proof, it suffices to consider the linearization of the system \eqref{eq:SRF-PLL} at equilibria $\delta^k\in\mathcal{E}_\star^{-}$ and evaluate the sign of the resulting eigenvalues.
\end{proof}\vspace{0.05cm}

    \begin{proposition}\label{prop:SRF-RoA} Any element of $\mathcal{E}_\star^+$ is a locally asymptotically stable equilibrium of the system \eqref{eq:SRF-PLL} and an inner estimate of the corresponding region of attraction (RoA) $\mathcal{R}_1^h\subset\mathbb{R}^2$, $h\in 2\mathbb{Z}$, is given by
    \begin{equation}\label{eq:RoA-SRF-PLL}\resizebox{1\hsize}{!}{$
    \hat{\mathcal{R}}_{1}^h :=\left\{(\delta,\hat{\omega})\in\mathbb{R}^2:\;\left(1-\cos(\delta-h \pi)\right)+\frac{(\hat{\omega}-\omega)^2}{2\gamma V}<\frac{K_I}{\gamma V}\right\}.$}
    \end{equation}
    \end{proposition}\vspace{0.05cm}
    
    \begin{proof}
Consider the Lyapunov functions
    \begin{equation}
        \mathcal W_1^h (\delta,\hat \omega):=\gamma V(1-\cos(\delta- h\pi ))+\frac{1}{2K_I}(\hat{\omega}-\omega)^2,
    \end{equation}
     which are all positive definite in $[-\pi\;\pi]\times\mathbb{R}$, modulo $2\pi$, with respect to the corresponding equilibrium point $(h\pi,\omega)$. By computing its derivative along the system's trajectories, we get:
    \begin{equation}
\begin{aligned}
    \dot{\mathcal W}_1^h =&-K_P\gamma^2 V^2\sin^2(\delta- h\pi )+\gamma V\sin(\delta- h\pi )(\hat{\omega}-\omega)\\
    &+(\hat{\omega}-\omega)(-\gamma V\sin(\delta- h\pi  ))\\
    =&-K_P\gamma^2 V^2\sin^2(\delta- h\pi )\\
    \leq&0,
\end{aligned}
\end{equation}
    which suffices to prove boundedness of the trajectories. Local asymptotic stability can be then inferred using LaSalle invariance's principle since $\sin^2(\delta)=0$ implies $(\delta,\hat{\omega})=( h\pi ,\omega)$. An estimation of the RoA can be then computed taking the largest, compact sublevel set of $\mathcal{W}_1^h $.
    \end{proof}\vspace{0.05cm}

Recalling \eqref{eq:map-equilibria}, from Corollary \ref{prop:SRF-PLL-dq} and Proposition \ref{prop:SRF-linearized} we can then conclude that for any initial condition belonging to the set $\{\mathcal{V}\times\mathbb{R}\}/\mathcal{E}_{dq\star}^{-}$ the solutions of \eqref{eq:sys-dq}-\eqref{eq:SRF-PLL-dq} converge to $\mathcal{E}_{dq\star}^{+}$. Or, equivalently, that for any initial condition belonging to the set $\mathbb{R}^2/\mathcal{E}_{\star}^{-}$, the solutions of \eqref{eq:SRF-PLL} converge to a point of $\mathcal{E}_\star^+$. Hence, we say that for \textit{almost any}  initial conditions a correct estimation of the actual value of the angle and of the frequency is guaranteed. Equilibria and convergence properties of the incremental, scaled version of the SRF-PLL are illustrated in Fig. \ref{fig:plane} and Fig. \ref{fig:cylinder}, for different initial guesses and both in polar and $dq$ coordinates.
\begin{figure}
    \centering
    \includegraphics[width=0.5 \textwidth]{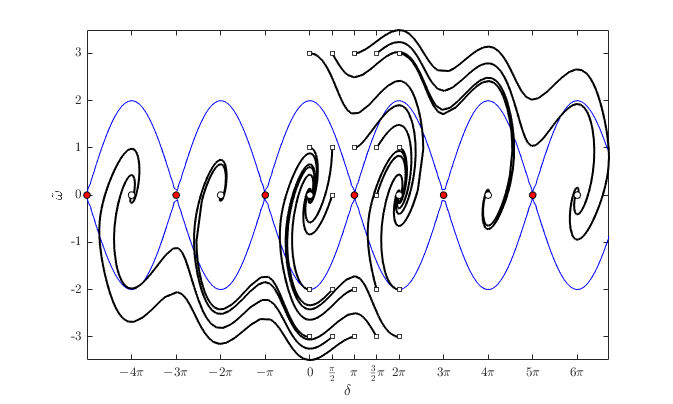}
    \caption{Solutions of the SRF-PLL dynamics in polar coordinates \eqref{eq:SRF-PLL}, for different initial guesses (square points) of the estimator. White and red circles denote respectively equilibria of $\mathcal{E}_\star^+$ and $\mathcal{E}_\star^-$, with blue lines denoting lower estimations of the RoAs of $\mathcal{E}_\star^+$.}
    \label{fig:plane}
\end{figure}
\begin{figure}
    \centering
    \includegraphics[width=0.5 \textwidth]{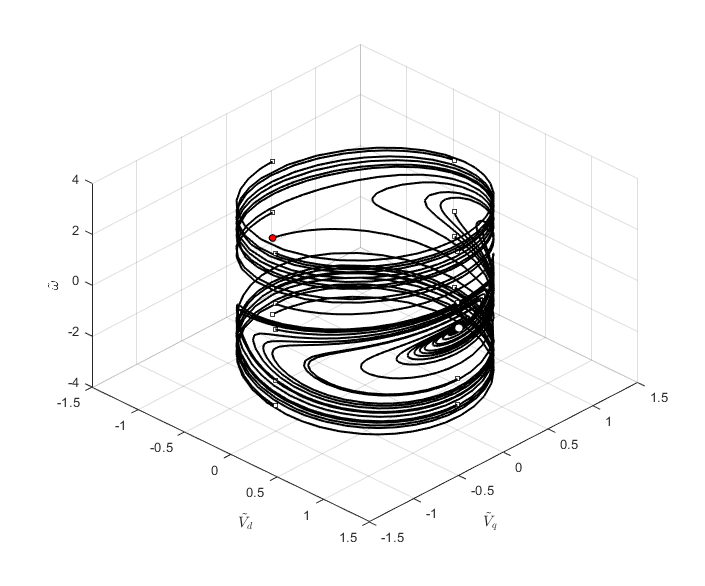}
    \caption{Solutions of the SRF-PLL dynamics in $dq$ coordinates \eqref{eq:sys-dq}-\eqref{eq:SRF-PLL-dq}, for different initial guesses (square points) of the estimator. White and red circles denote respectively the equilibria $(-v_{dq\star},\omega)$ and $(v_{dq\star},\omega)$. }
    \label{fig:cylinder}
\end{figure}
    An important observation is that the less accurate are the initial guesses, the poorest are the performances of the estimator. In polar coordinates, this is clearly captured by the fact that for some initial conditions the system's trajectories do not converge to the closest equilibrium, but rather exhibit a transients oscillating behavior before converging to farther equilibria. This suggests that information about the RoA of the equilibria of $\mathcal{E}_\star^{+}$, given by Proposition \ref{prop:SRF-RoA}, might be helpful to assess performances of the SRF-PLL, see also Remark~\ref{rem:RoA}.\vspace{0.05cm}

\section{ATAN-PLL}\label{sec:ATAN-PLL}
An alternative, widely diffused synchronization scheme for the estimation of the angle and of the frequency of \eqref{eq:voltage} relies on the use of a suitably defined inverse tangent function. We have first the following preliminary proposition. \vspace{0.05cm}

\begin{proposition}\label{prop:passivity-2}  Consider the system \eqref{eq:sys-dq}, let $v_{dq } \in\mathcal{V}$ and define
\begin{equation}\label{eq:incremental-2}
\tilde{u}:=\hat{u}-\omega,\quad \tilde y_2:=\atantwo(\hat V_q,\hat V_d)-\atantwo( V_{q },V_{d }).
\end{equation}
 Then, the map $\tilde{u}\rightarrow\tilde y_2$ is passive with storage function
 \begin{equation}
      \mathcal H_{2} (\hat v_{dq})=\frac{1}{2}\left[\atantwo(\hat V_q,\hat V_d)-\atantwo( V_{q },V_{d })\right]^2 
  \end{equation}
\end{proposition}\vspace{0.05cm}

\begin{proof}
The proof is similar to the proof of Proposition~\ref{prop:passivity-1}. Note that $\mathcal{H}_2=\frac{1}{2}\tilde y_2^2$. By calculating its derivative along the system's trajectories we get:
\begin{equation}\resizebox{1\hsize}{!}{$
    \dot{\mathcal{H}}_2=\left(-\tilde y_2\frac{\hat V_q}{\gamma^2 V^2}\right)\left( -\hat V_q \tilde{u}\right)+\left(\tilde y_2\frac{\hat V_d}{\gamma^2 V^2}\right)\left( \hat V_d \tilde{u}\right)=\tilde y_2\tilde{u},$}
\end{equation}
where in the last equivalence we used the fact that $\hat v_{dq}\in\mathcal{V}$, thus completing the proof.
\end{proof}\vspace{0.05cm}

We can now determine a class of estimators, referred as generalized arctangent (gATAN)-PLL, that guarantee a solution to the synchronization problem, similarly to Proposition~\ref{prop:SRF-PLL-dq}.\vspace{0.05cm}

\begin{proposition}\label{prop:ATAN-PLL-dq} Let $v_{dq } \in\mathcal{V}$. Consider the system \eqref{eq:sys-dq} in closed-loop with the PLL:
\begin{equation}\label{eq:gATAN-PLL-dq}
    \begin{aligned}
    \hat{u}&=-K_P\Phi(\tilde y_2)+\hat{\omega}\\
    \dot{\hat{\omega}}  &=-K_I\tilde y_2,
    \end{aligned}
\end{equation}
with state $\hat{\omega}\in\mathbb{R}$, $\tilde y_2$ given by \eqref{eq:incremental-2}, $K_I$, $K_P$ positive scalars and $\Phi$ is a scalar function verifying 
$\Phi(0)=0$, $\Phi(s)s>0$ for any $s\neq 0$. Then, the equilibrium point $(v_{dq },\omega)\in\mathcal{V}\times\mathbb{R}$ is globally asymptotically stable.
\end{proposition}\vspace{0.05cm}

\begin{proof}
Consider the energy function
  \begin{equation}
      \mathcal W_2 (\hat v_{dq},\hat{\omega})=\mathcal H_2(\hat v_{dq})+\frac{1}{2K_I}(\hat{\omega}-\omega)^2,
  \end{equation}
  which is positive definite and radially unbounded. We have then that
 \begin{equation}
 \begin{aligned}
      \dot{\mathcal{W}}_2 =&\dot{\mathcal H}_2+\frac{\hat{\omega}-\omega}{K_I}\dot{\hat{\omega}}\\
      =&\tilde y_2\left(-K_P\Phi(\tilde y_2)+\hat{\omega}-\omega\right)+\frac{\hat{\omega}-\omega}{K_I}\left(-K_I\tilde y_2\right)\\
        =&-K_P\Phi(\tilde y_2)\tilde y_2\\
      &\leq 0.
      \end{aligned}
      \end{equation}
      Global asymptotic stability follows using LaSalle invariance principle, noting that $\tilde y_2=0$ implies $\hat v_{dq}=v_{dq}$.
\end{proof}

A specific implementation of the controller \eqref{eq:ATAN-PLL-dq} for \mbox{$\Phi=\mathrm{Id}$} and   $v_{dq }=v_{dq\star}\in\mathcal{V}$ is given by the conventional ATAN-PLL~\cite{miskovic2018linear}, as stated in the following corollary.\vspace{0.05cm}    

 \begin{corollary}\label{corollary:ATAN-PLL}  Consider the system \eqref{eq:sys-dq} in closed-loop with the ATAN-PLL:
 \begin{equation}\label{eq:ATAN-PLL-dq}
\begin{aligned}
    \hat{u}&=-K_P\atantwo(\hat V_q,\hat V_d)+\hat{\omega}\\
    \dot{\hat{\omega}} &=-K_I\atantwo(\hat V_q,\hat V_d),
    \end{aligned}
\end{equation}
with $K_I$, $K_P$ positive scalars. Then, the equilibrium point $(v_{dq\star},\omega)\in\mathcal{V}\times\mathbb{R}$ is globally asymptotically stable.
\end{corollary}\vspace{0.05cm}

An immediate consequence of Corollary~\ref{corollary:ATAN-PLL} is that for any initial guess of the estimator both the frequency and angle estimation errors converges to zero. Similar to the case of the SRF-PLL an accurate estimation of the RoA of the equilibria can be used to assess performances of the estimator. For this purpose, it is convenient to transform the system in polar coordinates. Recalling \eqref{eq:voltage-dq}, \eqref{eq:sys-error}-\eqref{eq:ATAN-PLL-dq} can be rewritten as
\begin{equation}\label{eq:ATAN-PLL}
    \begin{aligned}
    \dot{\delta}&=-K_P\atantwo(\gamma V\sin(\delta),\gamma V\cos(\delta))+\hat{\omega}-\omega\\
    \dot{\hat\omega}&=-K_I\atantwo(\gamma V\sin(\delta),\gamma V\cos(\delta)).
    \end{aligned}
\end{equation}
We are now ready for the following proposition.\vspace{0.05cm}

\begin{proposition}\label{prop:ATAN-PLL} The system \eqref{eq:ATAN-PLL} admits infinitely many locally asymptotically stable equilibria $(h\pi,\omega)$, with $h\in 2\mathbb{Z}$, and an inner estimate of the corresponding RoA $\mathcal{R}_2^h\subset\mathbb{R}^2$ is given by
\begin{equation}\label{eq:RoA-ATAN-PLL}
 \hat{\mathcal{R}}_{2}^h :=\left\{(\delta,\hat{\omega})\in\mathbb{R}^2:\;4(\delta- h\pi )^2+(\hat{\omega}-\omega)^2<4\pi^2K_I\right\}.
\end{equation}
\end{proposition}\vspace{0.05cm}

\begin{proof}
The equilibria of \eqref{eq:ATAN-PLL} are the solutions of
 \begin{equation}
      0=\hat\omega-\omega,\quad 0=\atantwo(\gamma V\sin(\delta),\gamma V\cos(\delta)),
      \end{equation}
 from which follows that there exist infinitely many equilibria $(h\pi,\omega)$, with  $h\in 2\mathbb{Z}$. Note now that for any interval $[-\pi\;\pi)$ modulo $2\pi$ we further have
 $$\tilde y_2=\atantwo(\gamma V\sin(\delta),\gamma V\cos(\delta))=\delta- h\pi.$$
 Define the energy functions
  \begin{equation}
      \mathcal W_2^h (\delta,\hat{\omega})=\frac{1}{2}(\delta- h\pi )^2+\frac{1}{2K_I}(\hat{\omega}-\omega)^2,
  \end{equation}
   which are all locally positive definite with respect to the corresponding equilibrium point $( h\pi ,\omega)$. We have then that
 \begin{equation}
 \begin{aligned}
      \dot{\mathcal{W}}_2^h &=(\delta- h\pi )\left(-K_P\phi(\tilde y_2)+\hat{\omega}-\omega\right)-(\hat{\omega}-\omega)  (\delta- h\pi )\\
      &=-K_P\Phi(\tilde y_2)\tilde y_2\\
      &\leq 0.
      \end{aligned}
      \end{equation}
      Local asymptotic stability follows using LaSalle invariance principle, noting that $\tilde y_2=0$ implies $\delta= h\pi $, $\hat{\omega}=\omega.$ An estimation of the RoA can be then computed taking the largest, compact sublevel set of $\mathcal{W}_2^h $.
\end{proof}

\begin{figure}
    \centering
    \includegraphics[width=0.5 \textwidth]{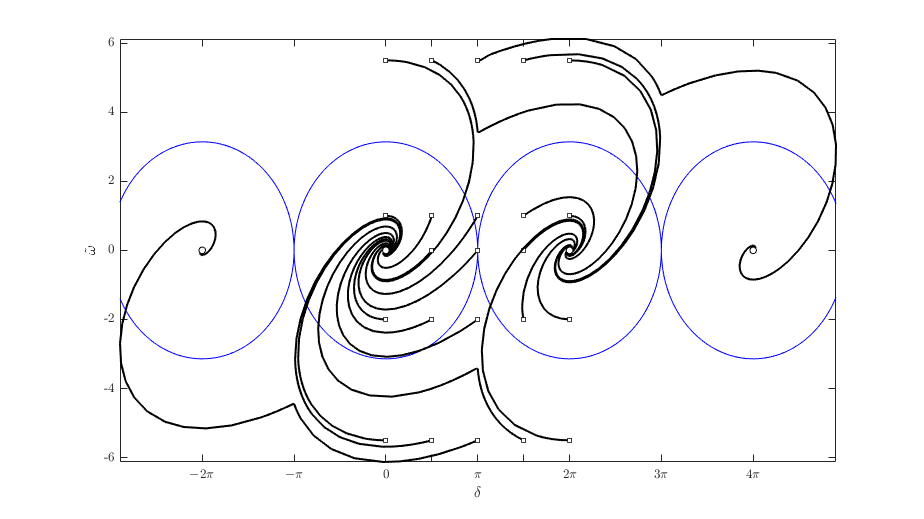}
    \caption{Solutions of the ATAN-PLL dynamics in polar coordinates \eqref{eq:ATAN-PLL}, for different initial guesses (square points) of the estimator. White circles denote the system equilibria,  while blue lines denote lower estimations of the corresponding RoAs.}
    \label{fig:ATAN-portrait}
\end{figure}
 
  Equilibria and convergence properties of the incremental, scaled version of the ATAN-PLL are illustrated in Fig.~\ref{fig:ATAN-portrait}.\vspace{0.05cm}
 
 \begin{remark}\label{rem:linear-ATAN} When restricted to the interval $[-\pi\;\pi)$, the ATAN-PLL estimation error dynamics \eqref{eq:ATAN-PLL} reduce to a linear system. This should be contrasted with the estimation error dynamics \eqref{eq:SRF-PLL}, which are nonlinear on the same interval.
 \end{remark}\vspace{0.05cm}
  \begin{remark}\label{rem:w-estimate}
    Note that convergence of \eqref{eq:gSRF-PLL-dq} or \eqref{eq:gATAN-PLL-dq} implies that both the signal $\hat u(t)$ and $\hat{\omega}(t)$ converge to the actual frequency $\omega$ and thus can be both interpreted as an estimate of such a signal. 
    \end{remark}\vspace{0.05cm}
 \begin{remark}\label{rem:RoA} By inspection of $\hat{\mathcal{R}}_1^h $ and $\hat{\mathcal{R}}_2^h $, it is immediate to see that larger values of the integral gain $K_I$ determine a larger RoA (or, at least, of its estimate) for both synchronization schemes. One should expect this to have also a relevant impact on the performances of the estimator, whose time-constant is defined indeed by $\tau:=1/K_I.$
\end{remark}
\section{Robustness to frequency variations}\label{sec:robustness}
In the previous section we developed a stability analysis of the SRF- and ATAN-PLL schemes under the assumption that the voltage at the PCC can be represented by a three-phase balanced, purely sinusoidal voltage source with a constant frequency. In this section we lift this last assumption, by extending such results to the case where the frequency is described via the following dynamics
\begin{equation}\label{eq:sys-omega}
\begin{aligned}
 \dot\omega(t)= \eta(t),
\end{aligned}
\end{equation}
where $\eta:\mathbb{R}_{\geq 0}\rightarrow\mathbb{R}$ is a function representing the rate of change of frequency (RoCoF). We now borrow a standard theorem from~\cite{khalilbook} to prove that, under the assumption that $\eta(t)$ is bounded, \textit{i.e.} $\vert \eta(t)\vert<\eta_\mathrm{max}$, the trajectories of the ATAN-PLL ultimately converge to a set that is shaped by the PLL gains (similar conclusions can be drawn for the SRF-PLL, a result that was obtained also in~\cite{rueda} using different arguments). To ease notation, let us define $\tilde{\omega}:=\hat{\omega}-\omega$. From \eqref{eq:ATAN-PLL} and \eqref{eq:sys-omega} we obtain:
\begin{equation}\label{eq:ATAN-PLL-omega}
    \begin{aligned}
    \dot{\delta}&=-K_P\atantwo(\gamma V\sin(\delta),\gamma V\cos(\delta))+\tilde{\omega}\\
    \dot{\tilde\omega}&=-K_I\atantwo(\gamma V\sin(\delta),\gamma V\cos(\delta))-\eta.
    \end{aligned}
\end{equation}
We have then the following proposition.\vspace{0.05cm}

\begin{proposition}\label{prop:robustness} For any initial condition $(\delta_0,\tilde{\omega}_{0})\in\mathcal{R}_2^0 $ the trajectories of the system \eqref{eq:ATAN-PLL-omega} are ultimately bounded,
with ultimate bound given by
\begin{equation}
\begin{aligned}
&\min_{\substack{\epsilon>0}} 
\frac{\sqrt{(\epsilon^2+K_I)\lambda_\mathrm{max}(P_\epsilon)}}{K_I\lambda_\mathrm{min}(P_\epsilon)\lambda_\mathrm{min}(Q_\epsilon)}\eta_\mathrm{max} \quad\mathrm{s.t.}\; P_\epsilon>0, Q_\epsilon>0,
\end{aligned}    
\end{equation}
with    \begin{equation}\label{eq:ATAN-PLL-P}
    \begin{aligned}
 P_\epsilon :=\begin{bmatrix}
    1 &-\frac{\epsilon}{2}\\
    -\frac{\epsilon}{2}& K_I
    \end{bmatrix},\;
 Q_\epsilon :=\begin{bmatrix}
    K_P-\epsilon K_I &-\frac{\epsilon}{2}K_P\\
    -\frac{\epsilon}{2}K_P& \epsilon
    \end{bmatrix}.
    \end{aligned}
\end{equation}
\end{proposition}
\begin{proof}
Consider the energy function
\begin{equation}
    \mathcal{W}_{2\epsilon}(\delta,\tilde\omega):=\frac{1}{2}\delta^2+\frac{1}{2K_I}\tilde{\omega}^2-\epsilon  \delta\frac{\tilde{\omega}}{K_I},
\end{equation}
defined in $(-\pi\;\pi)\times\mathbb{R}$, which is positive being $P_\epsilon>0$. By calculating the derivative of $\mathcal{W}_{2\epsilon}$ along the system's trajectories we get:
\begin{equation}
    \begin{aligned}
    \dot{\mathcal{W}}_{2\epsilon}=&\delta  \left(-K_P\delta+\hat{\omega}\right)+\frac{\hat{\omega}}{K_I}\left(-K_I\delta-\eta\right)+\\
    &-\epsilon \frac{\tilde{\omega}}{K_I}\left(-K_P\delta+\hat{\omega}\right)-\frac{\epsilon}{K_I}\delta\left(-K_I\delta-\eta\right)\\
    =&-\begin{bmatrix}
    \delta \\ \frac{\tilde{\omega}}{K_I}
    \end{bmatrix}^\top
   \begin{bmatrix}
    K_P-\epsilon K_I &-\frac{\epsilon}{2}K_P\\
    -\frac{\epsilon}{2}K_P& \epsilon
    \end{bmatrix}
    \begin{bmatrix}
    \delta \\ \frac{\tilde{\omega}}{K_I}
    \end{bmatrix}
   +\\
   &+\begin{bmatrix}
    \delta\\\frac{\tilde{\omega}}{K_I}
    \end{bmatrix}^\top \begin{bmatrix}
    \frac{\epsilon}{K_I}\\-1
    \end{bmatrix}\eta(t)
    \end{aligned}
\end{equation}
where in the first equivalence we substituted the estimation error dynamics, while the last one follows from \eqref{eq:ATAN-PLL-P}. Using Schur's complement it is then possible to show that there exists always an $\epsilon>0$ such that $P_\epsilon$ and $Q_\epsilon$ are positive definite. Hence $Q_\epsilon$ can be made positive so that the quadratic term eventually dominates the linear one, while preserving positive definiteness of $\mathcal{W}_{2\epsilon}$. More precisely, using standard norm inequalities and recalling that \mbox{$\vert \eta(t)\vert<\eta_\mathrm{max}$} we get:
\begin{equation}
    \dot{\mathcal{W}}_{\epsilon}<-\lambda_\mathrm{min}( Q_\epsilon)\|x\|^2+\frac{\sqrt{\epsilon^2+K_I^2}}{K_I}\eta_\mathrm{max}\|x\|,
\end{equation}
where $x:=(\delta,\tilde{\omega}/K_I)\in\mathbb{R}^2$, from which follows
\begin{equation}
    \dot{\mathcal{W}}_{\epsilon}<0\quad \forall \|x\|>\frac{\sqrt{\epsilon^2+K_I^2}}{\lambda_\mathrm{min}(Q)K_I}\eta_\mathrm{max}.
\end{equation}
The ultimate bound can be then obtained using Theorem 4.18 in~\cite{khalilbook}.
\end{proof}\vspace{0.05cm}

\section{Simulations}\label{sec:simulations}
To validate our result via simulations we consider a controlled VSC interfaced to a sufficiently strong, high-voltage AC grid, in accordance with Assumption~\ref{ass:2}. A conventional VSC current control scheme that guarantees fast regulation of active and reactive power under limited voltage variations is further considered, to comply with Assumption~\ref{ass:1}. Physical parameters and values of the reference signals are given in Table~\ref{tab:parameters}.
\begin{table}[t]
    \caption{Physical parameters and nominal operating points.}\label{tab:parameters}
    \centering
    \begin{tabularx}{\textwidth}{ll|ll|ll}
     $L$ & $65.20\;\textrm{mH}$ & $L_g$ & $32.60 \;\textrm{mH}$ & $P_\mathrm{nom}$ & $\phantom{00}1\;\textrm{GW}$    \\
     $R$ & $\phantom{0}1.02\;\Omega$ & $R_g$ & $\phantom{0}0.20\;\Omega$ & $V_\mathrm{nom}$ & $320\;\textrm{kV}$
     \end{tabularx} 
\end{table}
Based on this setup, we first illustrate performances and stability properties of both traditional and generalized ATAN-PLL under the assumption that the grid has relatively high inertia.\footnote{Similar results, not reported here for sake of brevity, apply to the (g)SRF-PLL.} More precisely, we consider step variations in the VSC power references over a time span of $0.5\;s$, inducing any $T=0.1\;s$ almost instantaneous variation of the phase, with the frequency remaining constant. In Fig.~\ref{fig:sim-1} we show active power responses resulting from synchronization of the VSC via two differently tuned ATAN-PLLs \eqref{eq:ATAN-PLL-dq}---referred as ATAN-PLL$_1$ and ATAN-PLL$_{10}$---and a gATAN-PLL \eqref{eq:gATAN-PLL-dq}. For all schemes we set the integral gain to $K_I=10^{3}$, which guarantees that the estimated RoA given by~\eqref{eq:RoA-ATAN-PLL} is sufficiently large for any of the operating points of interest. For the ATAN-PLLs we select the proportional gains $K_{P1}=2\cdot 10^{2}$ and $K_{P10}=10K_{P1}$. For the gATAN-PLL, we select $\Phi$ as an appropriate piecewise proportional function, allowing to adapt the proportional gain as a function of the output $\tilde y_2$. It is shown that all schemes guarantee a stable behavior and are able to ensure both the required power flow and a correct estimation of the frequency. It can be further seen that a suitable design of adaptive gains via the gATAN-PLL can improve the PLL performances.
\begin{figure}
    \centering
    \includegraphics[width=0.5 \textwidth]{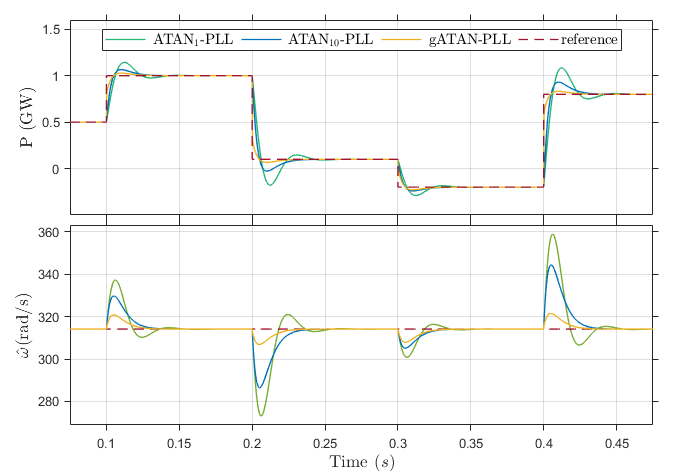}
    \caption{Active power and frequency responses under ATAN-PLL$_1$, ATAN-PLL$_{10}$ and gATAN-PLL, for a grid with relatively high-inertia.}
    \label{fig:sim-1}
\end{figure}

As a second scenario, we assume that the grid is characterized by a low inertia and that an inadvertent change in the power demand generates a deviation of the frequency from its nominal value. More precisely, similar to~\cite{rueda}, we suppose that the frequency is constant and equal to $f:=50 \;Hz$ during the time interval from $0$ to $1\;s$ and immediately afterwards the following additive disturbance is considered:
\begin{equation}
    \Delta\omega(t):=-8\pi e^{-0.1(t-1)}\sin(0.2(t-1)).
\end{equation}
Frequency estimates under SRF- and ATAN-PLL are illustrated in Fig.~\ref{fig:sim-2}, where it is shown that they have an almost identical response, which converges to the actual value of the frequency after approximately $7\;s$, in accordance with the result of Proposition~\ref{prop:robustness}.
\begin{figure}
    \centering
    \includegraphics[width=0.5 \textwidth]{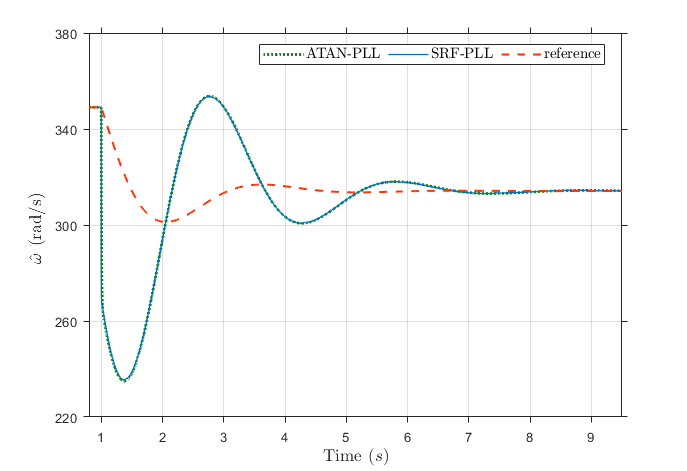}
    \caption{Frequency responses under SRF-PLL and ATAN-PLL under bounded variation of the frequency.}
    \label{fig:sim-2}
\end{figure}
\section{Conclusions}\label{sec:conclusions}
In this paper we constructively derived, using passivity arguments, two classes of PLL synchronization mechanisms for grid-connected VSC that ensure a correct estimation of the angle and of the frequency at the PCC. It is shown that (almost) global synchronization can be obtained for both classes of PLLs using representations of the system in polar and $dq$ coordinates, and that the conventional SRF- and ATAN-PLL correspond to particular instances of such generalized synchronization schemes. Finally, we have shown that these schemes are robust to variations of the frequency---a scenario particularly relevant for VSCs connected to grids with low inertia. Future investigation will be focused on the extension of such result to the case of weak grids, \textit{i.e.} grid characterized by a low SCR, by lifting Assumption~\ref{ass:1} and Assumption~\ref{ass:2}.

\bibliography{bibliografia}

\begin{thebibliography}{10}

\bibitem{bose}
B.~K. Bose, ``Global warming: Energy, environmental pollution, and the impact
  of power electronics,'' {\em IEEE Industrial Electronics Magazine}, vol.~4,
  no.~1, pp.~6--17, 2010.

\bibitem{chung2000phase}
S.-K. Chung, ``Phase-locked loop for grid-connected three-phase power
  conversion systems,'' {\em IEE Proceedings-Electric Power Applications},
  vol.~147, no.~3, pp.~213--219, 2000.

\bibitem{teodorescu}
R.~Teodorescu, M.~Liserre, and P.~Rodríguez, {\em Grid Converters for
  Photovoltaic and Wind Power Systems}.
\newblock John Wiley and Sons, 2011.

\bibitem{rantzer}
A.~Rantzer, ``Almost global stability of phase-locked loops,'' in {\em
  Proceedings of the 40th IEEE Conference on Decision and Control (Cat.
  No.01CH37228)}, vol.~1, pp.~899--900 vol.1, 2001.

\bibitem{abramovitch}
D.~Y. Abramovitch, ``Phase-locked loops: a control centric tutorial,'' {\em
  Proceedings of the 2002 American Control Conference (IEEE Cat. No.CH37301)},
  vol.~1, pp.~1--15 vol.1, 2002.

\bibitem{molinasPLL}
J.~A. Suul, S.~D'Arco, P.~Rodriguez, and M.~Molinas, ``Extended stability range
  of weak grids with voltage source converters through impedance-conditioned
  grid synchronization,'' in {\em 11th IET International Conference on AC and
  DC Power Transmission}, pp.~1--10, Feb 2015.

\bibitem{karimi-weak}
M.~Z. Mansour, S.~P. Me, S.~Hadavi, B.~Badrazadeh, A.~Karimi, and B.~Bahrani,
  ``Nonlinear transient stability analysis of phase-locked loop based
  grid-following voltage source converters using {L}yapunov's direct method,''
  {\em IEEE Journal of Emerging and Selected Topics in Power Electronics},
  pp.~1--1, 2021.

\bibitem{sun2019impact}
Y.~Sun, E.~De~Jong, X.~Wang, D.~Yang, F.~Blaabjerg, V.~Cuk, and J.~Cobben,
  ``The impact of {PLL} dynamics on the low inertia power grid: A case study of
  bonaire island power system,'' {\em Energies}, vol.~12, no.~7, p.~1259, 2019.

\bibitem{rueda}
J.~G. Rueda-Escobedo, S.~Tang, and J.~Schiffer, ``A performance comparison of
  {PLL} implementations in low-inertia power systems using an observer-based
  framework,'' {\em IFAC-PapersOnLine}, vol.~53, no.~2, pp.~12244--12250, 2020.
\newblock 21st IFAC World Congress.

\bibitem{VAN}
A.~J. V.~d. Schaft, {\em $\mathcal{L}_2$-Gain and Passivity Techniques in
  Nonlinear Control}.
\newblock Springer-Verlag New York, Inc., 2017.

\bibitem{miskovic2018linear}
V.~Miskovic, V.~Blasko, T.~M. Jahns, R.~D. Lorenz, and P.~M. Jorgensen,
  ``Linear phase-locked loop,'' in {\em 2018 IEEE Energy Conversion Congress
  and Exposition (ECCE)}, pp.~5677--5683, IEEE, 2018.

\bibitem{garces}
M.~Bravo, A.~Garcés, O.~D. Montoya, and C.~R. Baier, ``Nonlinear analysis for
  the three-phase {PLL}: A new look for a classical problem,'' in {\em 2018
  IEEE 19th Workshop on Control and Modeling for Power Electronics (COMPEL)},
  pp.~1--6, 2018.

\bibitem{khalilbook}
H.~Khalil, {\em Nonlinear Systems}.
\newblock Pearson Education, Prentice Hall, 2002.

\end{thebibliography}
\bibliographystyle{ieeetr}

\end{document}